%% file: topk.tex
\documentclass[11pt, letterpaper]{article}
\usepackage[margin=1in]{geometry}

\usepackage{amsfonts, amsmath, amssymb, amsthm}
\usepackage{booktabs}
\usepackage{color}
\usepackage{euscript}
\usepackage{graphicx}
\usepackage{microtype}
\usepackage{multirow}
\usepackage{wrapfig}

\newtheorem{observation}{Observation}

\input{yf-def.tex}

\def\figcapup{\vspace{-3mm}}
\def\figcapdown{\vspace{-0mm}}
\def\extraspacing{\vspace{4mm} \noindent}

\def\vgap{\vspace{1mm}}

\def\euT{\EuScript{T}}
\def\nil{\mathit{nil}}

\def\pilot{\mathit{pilot}}

\def\score{\mathit{score}}

\def\G{\textrm{\bm$G$}}

\begin{document}
\begin{sloppy}

\title{A Dynamic I/O-Efficient Structure\\ for One-Dimensional Top-k Range Reporting\thanks{This paper supersedes an earlier version on arXiv with the title ``On Top-k Search and Range Reporting''.}}

\author{
    Yufei Tao \\[2mm]
    Department of Computer Science and Engineering \\
	Chinese University of Hong Kong \\
	Hong Kong \\
	{\em taoyf@cse.cuhk.edu.hk} 
}

\date{}

\maketitle
\vspace{15mm}
\begin{abstract}
	We present a structure in external memory for {\em top-$k$ range reporting}, which uses linear space, answers a query in $O(\lg_B n + k/B)$ I/Os, and supports an update in $O(\lg_B n)$ amortized I/Os, where $n$ is the input size, and $B$ is the block size. This improves the state of the art which incurs $O(\lg^2_B n)$ amortized I/Os per update. 
\end{abstract}

\vspace{80mm}

\thispagestyle{empty}
\setcounter{page}{0}
\pagebreak

\section{Introduction} \label{sec:intro}

In the {\em top-$k$ range reporting problem}, the input is a set $S$ of $n$ points in $\real$, where each point $e \in S$ carries a distinct\footnote{This is a standard assumption \cite{abz11, st12} to guarantee the uniqueness of a top-$k$ result. See \cite{st12} for two semantic extensions to remove the assumption, and how to reduce those extensions to the standard top-$k$ problem with distinct weights.} real-valued {\em score}, denoted as $\score(e)$. Given an interval $q = [x_1, x_2]$ and an integer $k$, a query returns the $k$ points in $S(q) = S \intr q$ with the highest scores. If $|S(q)| < k$, the entire $S(q)$ should be returned. The goal is to store $S$ in a structure so that queries can be answered efficiently.

\extraspacing {\bf Motivation.} Top-$k$ search in general is widely acknowledged as an important operation in a large variety of information systems (see an excellent survey \cite{ibs08}). It plays a central role in applications where an end user wants only a small number of elements with the best competitive quality, as opposed to all the elements satisfying a query predicate. Top-$k$ range reporting---being an extension of classic range reporting---is one of the most fundamental forms of top-$k$ search. A representative query on a hotel database is ``{\em find the 10 best-rated hotels whose prices are between 100 and 200 dollars per night}''. Here, each point $e \in S$ represents the price of a hotel, with $\score(e)$ corresponding to the hotel's user rating. In fact, queries like the above are so popular that database systems nowadays strive to make them first-class citizens with direct algorithm support. This calls for a space-economic structure that can guarantee attractive query and update efficiency. 

\extraspacing {\bf Computation Model.} We study the problem in the {\em external memory} (EM) model \cite{av88}. A machine is equipped with $M$ words of memory, and a disk of unbounded size that has been formatted into {\em blocks} of size $B$ words. An I/O either reads a block of data from the disk to memory, or conversely, writes $B$ words in memory to a disk block. The {\em space} of a structure is the number of blocks it occupies, whereas the {\em time} of an algorithm is the number of I/Os it performs. CPU calculation is for free. A word has  $\Omega(\lg n)$ bits, where $n \ge B$ is the input size of the problem at hand. The values of $M$ and $B$ satisfy the condition $M = \Omega(B)$.\footnote{$M$ can be as small as $2B$ in the model defined in \cite{av88}. However, any algorithm that works on $M = cB$ with constant $c > 2$ can be adapted to work on $M = 2B$ with only a constant blowup in space and time. Therefore, one might as well consider that $M = \Omega(B)$.} 

\vgap 

Throughout this paper, a space/time complexity holds in the worst case by default. A logarithm $\lg_b x$ is defined as $\max\{1, \log_b x\}$, and  $b = 2$ if omitted. {\em Linear} cost should be understood as $O(n/B)$ whereas {\em logarithmic} cost as $O(\lg_B n)$.

\subsection{Previous Work} \label{sec:intro-pre}

Top-$k$ range reporting was first studied by Afshani, Brodal and Zeh \cite{abz11}, who gave a static structure of $O(n/B)$ space that answers a query in $O(\lg_B n + k/B)$ I/Os. The query cost is optimal, as can be shown via a reduction from predecessor search \cite{pt06}. They also analyzed the space-query tradeoff for an {\em ordered} variant of the problem, where the top-$k$ elements need to be sorted  by score. Their result suggests that when the space usage is linear, one can achieve nearly the best query efficiency by simply solving the unordered version in $O(\lg_B n + k/B)$ I/Os, and then sorting the retrieved elements (see \cite{st12} for more details). For the unordered version, Sheng and Tao \cite{st12} proposed a dynamic structure that has the same space and query cost as \cite{abz11}, but supports an update in $O(\lg^2_B n)$ amortized I/Os.

\vgap

In internal memory, by combining a priority search tree \cite{m85} and Frederickson's selection algorithm \cite{f93} on heaps, one can obtain a pointer-machine structure that uses $O(n)$ words, answers a query in $O(\lg n + k)$ time, and supports an update in $O(\lg n)$ time. In RAM, Brodal, Fagerberg, Greve and Lopez-Ortiz \cite{bfgl09} considered a special instance of the problem where the input points of $S$ are from the domain $[1, n]$. They gave a linear-size structure with $O(1+k)$ query time (which holds also for the ordered version). 

\vgap 

It is worth mentioning that top-$k$ search has received considerable attention in many other contexts. We refer the interested readers to recent works \cite{mn12, sstv13} for entry points into the literature. 

\subsection{Our Results} \label{sec:intro-ours}

We improve the state of the art \cite{st12} by presenting a new structure with logarithmic update cost:

\begin{theorem} \label{thm:main-final}
    For top-$k$ range reporting, there is a structure of $O(n/B)$ space that answers a query in $O(\lg_B n + k/B)$ I/Os, and supports an insertion and a deletion in $O(\lg_B n)$ I/Os amortized.
\end{theorem}

We achieve logarithmic updates by combining three methods. The first one adapts the aforementioned pointer machine structure---which combines a priority search tree with Frederickson's heap selection algorithm---to external memory. This gives a linear-size structure that can be updated in $O(\lg_B n)$ amortized I/Os, but answers a query in $O(\lg n + k/B)$ I/Os (note that the log base is 2). We use the structure to handle $k \ge B \lg n$ in which case its query cost is $O(\lg n + k/B) = O(k/B)$.

\vgap

The second method applies directly the structure of \cite{st12}. Looking at their analysis carefully, one sees that their amortized update cost is in fact $O(\lg_B n + \fr{\lg n}{B^{1/6}} \lg_B n)$. In other words, when $\lg n \le B^{1/6}$, the structure already achieves logarithmic update cost. 

\vgap

The most difficult case arises when $\lg n > B^{1/6}$, or equivalently, $B < \lg^6 n$. We observe that, since $k \ge B \lg n$ has already been taken care of, it remains to target $k < B \lg n < \lg^7 n$. Motivated by this, we develop a linear-size structure that can be updated in $O(\lg_B n)$ I/Os, and answers a query with $k = O(\polylg n)$ in $O(\lg_B n + k/B)$ I/Os. The most crucial idea behind this structure is to use a suite of ``RAM-reminiscent'' techniques to unleash the power of manipulating individual bits. 

\vgap 

Theorem~\ref{thm:main-final} can now be established by putting together the above three structures using standard global rebuilding techniques.

\section{A Structure for \bm{\large ${k} = \Omega(B \lg n)$}} \label{sec:bigk}

In this section, we will prove:

\begin{lemma} \label{lmm:bigk-main}
    For top-$k$ range reporting, there is a structure of $O(n/B)$ space that answers a query in $O(\lg n + k/B)$ I/Os, and supports an insertion and a deletion in $O(\lg_B n)$ I/Os amortized.
\end{lemma}

Top-$k$ range reporting has a geometric interpretation. We can convert $S$ to a set $P$ of points, by mapping each element $e \in S$ to a 2d point $(e, \score(e))$. Then, a top-$k$ query with $q = [x_1, x_2]$ equivalently reports the $k$ highest points of $P$ in the vertical slab $q \times (-\infty, \infty)$. This is the perspective we will take to prove Lemma~\ref{lmm:bigk-main}. 

\vgap 

Our structure is essentially an {\em external priority search tree} \cite{asv99} on $P$ with a constant fanout. However, we make two contributions. First, we develop an algorithm using this structure to answer top-$k$ range queries. Second, we explain how to update the structure in $O(\lg_B n)$ I/Os. Note that an update by the standard algorithm of \cite{asv99} requires $O(\lg n)$ I/Os.  

\begin{figure}
	\begin{center}
		\includegraphics[height=40mm]{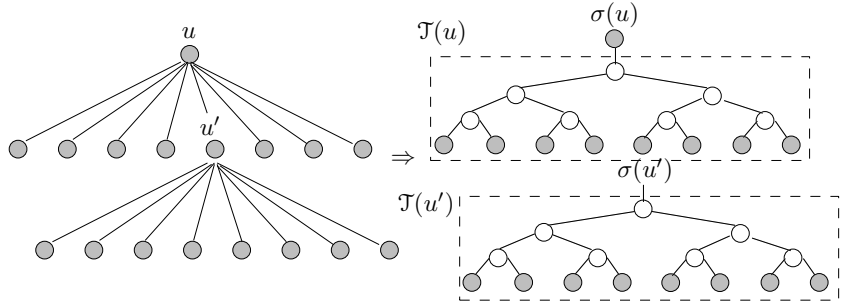}
	\end{center}
	\figcapup
	\caption{Concatenating secondary binary trees}
	\label{fig:bigk-concat}
	\figcapdown
\end{figure}

\extraspacing {\bf Structure.} Let $T$ be a {\em weight balanced B-tree} (WBB-tree) \cite{av03} on the x-coordinates of the points in $P$. The leaf capacity and branching parameter of $T$ are both set to $B$. We number the levels of $T$ bottom up, with the leaves at level 0. For each node $u$ in $T$, we use $P(u)$ to denote the set of points whose x-coordinates are stored in the subtree of $u$. As a property of the WBB-tree, if $u$ is at level $i$, then $|P(u)|$ falls between $B^{i+1}/4$ and $B^{i+1}$; if $|P(u)|$ is outside this range, $u$ becomes {\em unbalanced} and needs to be remedied.

\vgap

Each node $u$ naturally corresponds to a vertical {\em slab} $\sigma(u)$ with $P(u) = \sigma(u) \cap P$.\footnote{Precisely, the slab of a leaf node $u$ is $[x, x') \times (-\infty, \infty)$ where $x$ is the smallest x-coordinate stored at $u$, and $x'$ is the smallest x-coordinate in the leaf node $u'$ succeeding $u$. If $u'$ does not exist, $x' = \infty$. The slab of an internal node unions those of all its child nodes.} Let $u, u'$ be child nodes of the same parent. We say that $u'$ is a {\em right sibling} of $u$ if $\sigma(u')$ is to the right of $\sigma(u)$. Otherwise, $u'$ is a {\em left sibling} of $u$. Note that a node can have multiple left/right siblings, or none (if it is already the left/right most child). 

\vgap

Consider now $u$ as an internal node with child nodes $u_1, ..., u_f$ where $f = O(B)$ (we always follow the left-to-right order in listing out child nodes). We associate $u$ with a binary search tree $\euT(u)$ of $f$ leaves, which correspond to $\sigma(u_1), ..., \sigma(u_f)$, respectively. Let $v$ be an internal node in $\euT(u)$. We define $\sigma(v) = \cup_{j=j_1}^{j_2} \sigma(u_j)$, where $\sigma(u_{j_1}), \sigma(u_{j_1+1}), ..., \sigma(u_{j_2})$ are the leaves of $\euT(u)$ below $v$, and accordingly, define $P(v) = \sigma(v) \intr P$.  

\vgap 

Notice that we can view $T$ insteads as one big tree {\bm${T}$} that concatenates the secondary binary trees of all the nodes in $T$. Specifically, if $u'$ is a child of $u$ in $T$, the concatenation makes the root of $\euT(u')$ the only child of the leaf $\sigma(u')$ of $\euT(u)$. See Figure~\ref{fig:bigk-concat}. {\bm${T}$} is almost a binary tree except that some internal nodes have only one child which is an internal node itself. However, this is only a minor oddity because any path in {\bm$T$} of 3 nodes must contain at least one node with two children. The height of {\bm$T$} is $O(\lg n)$. 

\vgap

Each node $v$ in {\bm$T$} is associated with a set---denoted as $\pilot(v)$---of {\em pilot points} satisfying two conditions:
\begin{itemize} 
    \item The points of $\pilot(v)$ are the highest among all points $p \in P(v)$ that are not stored in any $\pilot(\hat{v})$, where $\hat{v}$ is a proper ancestor of $v$ in {\bm$T$}.
    
    \item If less than $B/2$ points satisfy the above condition, $\pilot(v)$ includes all of them. Otherwise, $B/2 \le |\pilot(v)| \le 2B$. In any case, $\pilot(v)$ is stored in $O(1)$ blocks. 
\end{itemize}
The lowest point in $\pilot(v)$ is called the {\em representative} of $\pilot(v)$.

\vgap 

Finally, for each internal node $u$ in $T$, we collect the representatives of the pilot sets of all the nodes in $\euT(u)$, and store these $O(B)$ representatives in $O(1)$ blocks---referred to as the {\em representative blocks} of $u$. 

\extraspacing {\bf Query.} Given a top-$k$ query with range $q = [x_1, x_2]$, we descend two root-to-leaf paths $\pi_1$ and $\pi_2$ in {\bm$T$} to reach the leaf nodes $z_1$ and $z_2$ whose slabs' x-ranges cover $x_1$ and $x_2$, respectively. In $O(\lg n)$ I/Os, we retrieve all the $O(B \lg n)$ pilot points of the nodes on $\pi_1 \cup \pi_2$, and eliminate those outside $q \times (-\infty, \infty)$. Let $Q_1$ be the set of remaining points.

\begin{figure}
	\begin{center}
		\includegraphics[height=50mm]{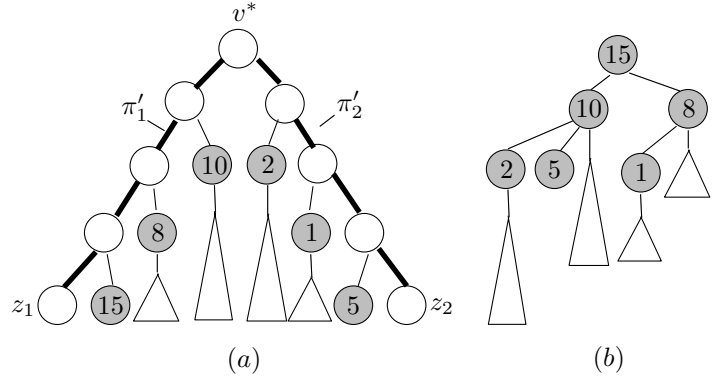}
	\end{center}
	\figcapup
	\caption{The gray nodes in Figure (a) constitute set $\Pi$. Each number is a node's sorting key in the heap rooted at that node. Figure (b) shows $H$ after heap concatenation.}
	\label{fig:bigk-heap}
	\figcapdown
\end{figure}

\vgap

Let $v^*$ be the least common ancestor of $z_1$ and $z_2$. Define $\pi_1'$ ($\pi_2'$) as the path from $v^*$ to $z_1$ ($z_2$). Let $\Pi$ be the set of nodes $v$ satisfying two conditions: 
\begin{itemize} 
    \item []
\begin{itemize}
    \item [(i)] $v \notin \pi_1' \cup \pi_2'$, but the parent of $v$ is in $\pi_1' \cup \pi_2'$;
    \item [(ii)] The x-range of $\sigma(v)$ is covered by $q$. 
\end{itemize}
\end{itemize}
For every such $v$, we can regard its subtree as a max-heap $H(v)$ as follows. First, $H(v)$ includes all the nodes $v'$ in the subtree of $v$ (in {\bm$T$}) with non-empty pilot sets. Second, the sorting key of $v'$ is the y-coordinate of the representative of $\pilot(v')$. In this way, we have identified at most $|\Pi|$ non-empty max-heaps, each rooted at a distinct node in $\Pi$. Concatenate these heaps into one, by organizing their roots into a binary max-heap based on the sorting keys of those roots. This can be done in $O(\lg n)$ I/Os\footnote{Using a linear-time ``make-heap'' algorithm; see \cite{clrs01}.}. Denote by $H$ the resulting max-heap after concatenation. See Figure~\ref{fig:bigk-heap}. 

\vgap

Set $\phi$ to a sufficiently large constant. We now invoke Frederickson's algorithm to extract the set $R$ of $\phi \cdot (\lg n + k/B)$ representatives in $H$ with the largest y-coordinates; this entails $O(\lg n + k/B)$ I/Os. Let $S_R$ be the set of nodes whose representatives are collected in $R$. Gather all the pilot points of the nodes of $S_R$ into a set $Q_2$. 

\vgap

Define a set $S_R^*$ of nodes as follows. For each node $v \in S_R$, we first add to $S_R^*$ all such siblings $v'$ of $v$ (in {\bm$T$}) that (i) $v' \notin S_R$, and (ii) the x-range of $\sigma(v')$ is contained in $q$. Second, if $v$ is an internal node, add all its child nodes in {\bm$T$} to $S_R^*$. Note that $|S_R^*| = O(|S_R|) = O(\lg n + k/B)$. We now collect the pilot points of all the nodes of $S_R^*$ into a set $Q_3$.   

\vgap 

At this moment, we have collected three sets $Q_1, Q_2, Q_3$ with a total size of $O(B\log n + k)$. We can now report the $k$ highest points in $Q_1 \cup Q_2 \cup Q_3$ in $O(\lg n + k/B)$ I/Os. The query algorithm performs $O(\lg n + k/B)$ I/Os in total. Its correctness is ensured by the fact below: 

\begin{lemma} 
    Setting $\phi = 16$ ensures that $Q_1 \cup Q_2 \cup Q_3$ includes the $k$ highest points in $q \times (-\infty, \infty)$.
\end{lemma}

\begin{proof} 
	We will focus on the scenario that the heap $H$ has at least $\phi \cdot (\lg n + k/B)$ representatives. Otherwise, $P$ has $O(B\lg n + k)$ points in $q \times (-\infty, \infty)$, and all of them are in $Q_1 \cup Q_2$ ($Q_3$ is empty).

	\vgap 
	
	We will first show that $|Q_1 \cup Q_2| \ge k$. This is very intuitive because $Q_1 \cup Q_2$ collects the contents of $\Omega(\lg n + k/B)$ pilot sets. However, a formal proof requires some effort because the pilot set of a node $v$ can have arbitrarily few points (in this case all the nodes in the proper subtree of $v$ must have empty pilot sets). We need a careful argument to address this issue. 
    
    \vgap 
    
	We say that a representative in $R$ is {\em poor} if its pilot set has less than $B/8$ points; otherwise, it is {\em rich}. Consider a poor representative $r$ in $R$; and suppose that it is a pilot point of node $v$, and its x-coordinate is stored in leaf node $z$. Note that $z$ stores the x-coordinates of at least $B/4$ points, all of which fall in $q$. By the fact that $r$ represents less than $B/8$ points, we know that at least $B/8$ points (with x-coordinates) in $z$ are pilot points of some proper ancestors of $v$ in {\bm$T$}, and therefore, appear in either $Q_1$ or $Q_2$. We  associate those $B/8$ points with $r$. On the other hand, we associate each rich representative with the at least $B/8$ points in its pilot set. 

	\vgap

	Thus, the $\phi \cdot (\lg n + k/B)$ representatives in $R$ are associated with at least $(\phi / 8) (B\lg n + k)$ points in $Q_1 \cup Q_2$. Each point $p \in Q_1 \cup Q_2$, on the other hand, can be associated with at most 2 representatives: the representative of the node where $p$ is a pilot point, and a poor representative whose x-coordinate is stored in the same leaf as $p$.\footnote{No two poor representatives can have their x-coordinates stored in the same leaf.} This implies $|Q_1 \cup Q_2| \ge (\phi/16)(B \lg n + k)$. Hence, $\phi = 16$ ensures $|Q_1 \cup Q_2| \ge k$. 
	
	\vgap 
	
	Finally, the inclusion of $Q_3$ ensures that no pilot point in $q \times (-\infty, \infty)$ but outside $Q_1 \cup Q_2 \cup Q_3$ can be higher than the lowest point in $Q_1 \cup Q_2$. The lemma then follows. 
\end{proof}

\extraspacing {\bf Insertion.} To insert a point $p$, first update the B-tree $T$ by inserting the x-coordinate of $p$. Let us assume for the time being that no rebalancing in $T$ is required. Then, we identify the node $v$ in {\bm$T$} whose pilot set should incorporate $p$. This can be achieved in $O(\lg_B n)$ I/Os by descending a single root-to-leaf path in $T$ (note: not {\bm$T$}), and inspect the representative blocks of the nodes on the path. We add $p$ to $\pilot(v)$. 

\vgap 

We say that a pilot set {\em overflows} if it has more than $2B$ points. If $\pilot(v)$ overflows, 
we carry out a {\em push-down} operation at $v$, which moves the $|\pilot(v)| - B$ lowest points of $\pilot(v)$ to the pilot sets of its at most 2 child nodes in {\bm$T$}. The resulting $\pilot(v)$ has size $B$. If the pilot set of a child $v'$ now overflows, we treat it in the same manner by performing a push-down at $v'$. We will analyze the cost of push-downs later.


\extraspacing {\bf Deletion.} To delete a point $p$, we identify the node $v$ in {\bm$T$} whose pilot set contains $p$. This can be done in $O(\lg_B n)$ I/Os by inspecting the representative blocks. We then remove $p$ from $\pilot(v)$. 

\vgap 

We say that a pilot set {\em underflows} if it has less than $B/2$ points, and yet, one of its child nodes has a non-empty pilot set. To remedy this, we define a {\em pull-up} operation at node $v'$ in {\bm$T$} as one that moves the $\min\{B/2, B - |\pilot(v')|\}$
highest points from 
\begin{eqnarray} 
    \bigcup_{\textrm{child $v''$ of $v'$ in {\bm$T$}}} \pilot(v'') \label{eqn:bigk-childunion} 
\end{eqnarray}
to $\pilot(v')$. If \eqref{eqn:bigk-childunion} has less than the requested number of points, the pull-up moves all the points of \eqref{eqn:bigk-childunion} into $\pilot(v)$, after which all proper descendants of $v'$ have empty subsets; we call such a pull-up a {\em draining} one. 

\vgap 

In general, if the pilot set of a node $v'$ underflows, we carry out at most two pull-ups at $v'$ until either $|\pilot(v')| = B$, or a draining pull-up has been performed. After the first pull-up, if the pilot set at a child node of $v'$ underflows, we should remedy that first (in the same manner recursively) before continuing with the second pull-up at $v'$. We will analyze the cost of pull-ups later. 

\vgap 

It is worth mentioning that we do not remove the x-coordinate of $p$ from the base tree $T$. This does not create a problem because we will rebuild the whole $T$ periodically, as clarified later. 


\extraspacing {\bf Rebalancing.} It remains to clarify how to rebalance $T$. Let $u^\*$ be the highest node in $T$ that becomes unbalanced after inserting $p$. Let $\hat{u}$ be the parent of $u^\*$. We rebuild the whole subtree of $\hat{u}$ in $T$, and the corresponding portion in {\bm$T$}. Let $l$ be the level of $\hat{u}$ in $T$. Our goal is to complete the reconstruction in $O(B^l)$ I/Os. A standard argument with the WBB-tree shows that every insertion accounts for $O(\lg_B n)$ I/Os of all the reconstructions. 

\vgap 

Let $\hat{v}$ be the root of $\euT(\hat{u})$. Essentially, we need to rebuild the subtree of $\hat{v}$ in {\bm$T$}, which has $O(B^l)$ nodes. The first step of our algorithm is to distribute all the pilot points stored in the subtree of $\hat{v}$ down to the leaves where their x-coordinates are stored, respectively. For this purpose, we simply push down all the pilot points of $\hat{v}$ to its child nodes in {\bm$T$}, and do so recursively at each child. We call this a {\em pilot grounding} process. 

\vgap 

We now reconstruct the subtrees of $\hat{u}$ and $\hat{v}$. First, it is standard to create all the nodes of $T$ in the subtree of $\hat{u}$, and all the nodes of {\bm$T$} in the subtree of $\hat{v}$ in $O(B^l)$ I/Os. What remains to do is to fill in the pilot sets. We do so in a bottom up manner. Suppose that we are to fill in the pilot set of $v$, knowing that the pilot sets of all the proper descendants of $v$ (in {\bm$T$}) have been computed properly. We populate $\pilot(v)$ using the same algorithm as treating a pilot set underflow at $v$. 

\vgap 

Next, we prove that the whole reconstruction takes $O(B^l)$ I/Os. Let us first analyze the pilot grounding process. We say that a {\em demotion event} occurs when a point moves from the pilot set of a parent node to that of a child. If $N_d$ represents the number of such events, we can bound the total cost of pilot grounding as $O(B^l + N_d / B)$. 

\vgap

To bound $N_d$, first consider a level-1 node $u$ in $T$. A node at level $j \ge 1$ of $\euT(u)$ triggers $O(jB)$ demotion events. Hence, the number of demotion events triggered by all the nodes of $\euT(u)$ is $\sum_{j=1}^{O(\lg B)} O(B/2^j)  \cdot O(jB) = O(B^2)$. As the subtree of $\hat{u}$ has $O(B^{l-1})$ level-1 nodes, they trigger $O(B^2 \cdot B^{l-1}) = O(B^{l+1})$ demotion events in total. 

\vgap

Now consider $u$ as a level-$i$ node of $T$ with $i \ge 2$. Each of the $B$ nodes in $\euT(u)$ can trigger $O(i B \lg B)$ demotion events, resulting in a total event count of $O(i B^2 \lg B)$ for $u$. Since there are $O(B^{l-i})$ nodes at level $i$, the number of demotion events due to the nodes from level 2 to level $l$ is at most
\begin{eqnarray}
    \sum_{i=2}^{l} O(B^{l-i}) \cdot O(iB^2 \lg B) = O(B^l \lg B). \nn
\end{eqnarray}
Therefore, $N_d = O(B^{l+1} + B^l \lg B) = O(B^{l+1})$. It follows that the pilot grounding process requires $O(B^l + N_d/B) = O(B^l)$ I/Os. 

\vgap 

The cost of filling pilot sets can be analyzed in the same fashion, by looking at {\em promotion events}---namely, a point moves from the pilot set of a child to that of the parent. If $N_p$ represents the number of such events, we can bound the cost of pilot set filling as $O(B^l + N_p/B)$. By an argument analogous to the one on $N_d$, one can derive that $N_p = O(B^{l+1})$. 

\extraspacing {\bf Push-Downs and Pull-Ups.} Next, we will prove that each update accounts for only $O(\fr{1}{B} \lg n)$ I/Os incurred by push-downs and pull-ups. At first glance, this is quite intuitive: inserting a point into a pilot set may ``edge out'' an existing point there to the next level of {\bm$T$}, which may then create a cascading effect every level down. Viewed this way, an insertion creates $O(\lg n)$ demotion events, and reversely, a deletion creates $O(\lg n)$ promotion events. As $\Omega(B)$ such events are handled by a push-down or pull-up using $O(1)$ I/Os, the cost amortized on an update should be $O(\fr{1}{B} \lg n)$. What complicates things, however, is the fact that pilot points may bounce up and down across different levels. Below we give an argument to account for this complication.  

\vgap 

We imagine some conceptual {\em tokens} that can be passed by a node to a child in {\bm$T$}, but never the opposite direction. Specifically, the {\em rules} for creating, passing, and deleting tokens are: 
\begin{enumerate} 
    \item When a point $p$ is being inserted into {\bm$T$}, we give $v$ an {\em insertion token} if $p$ is placed in $\pilot(v)$.
    \item When a point $p$ is deleted from {\bm$T$}, we give $v$ a {\em deletion token} if $p$ is removed from $\pilot(v)$. 
    \item In a push-down, when a point $p$ is moved from $\pilot(v)$ to $\pilot(v')$ (where $v'$ is a child of $v$), we take away an insertion token from $v$, and give it to $v'$. We will prove shortly that $v$ always has enough tokens to make this possible.
    
    \item In a pull-up, when a point $p$ is moved from $\pilot(v')$ to $\pilot(v)$ (where $v'$ is a child of $v$), we take away a deletion token from $v$, and give it to $v'$. Again, we will prove shortly that this is always do-able.
    \item When an insertion/deletion token reaches a leaf node, it disappears. 
    \item After a draining pull-up is performed at $v$, all the tokens in the subtree of $v$ disappear. 
    \item When the subtree of a node $v$ is reconstructed, all the tokens in the subtree disappear. 
\end{enumerate}

\begin{lemma} \label{lmm:bigk-inv}
    Our update algorithms enforce two invariants at all times: 
	\begin{itemize} 
		\item Invariant 1: every internal node $v$ in {\bm$T$} has at least $|\pilot(v)| - B$ insertion tokens. 
		\item Invariant 2: every internal node $v$ in {\bm$T$} has at least $B - |\pilot(v)|$ deletion tokens, unless all proper descendants of $v$ in {\bm$T$} have empty pilot sets.
	\end{itemize}
\end{lemma}

Notice that, by Invariant 1, a node $v$ with $|\pilot(v)| \le B$ is not required to hold any insertion tokens; likewise, by Invariant 2, a node $v$ with $|\pilot(v)| \ge B$ is not required to hold any deletion tokens. Furthermore, the two invariants ensure that the token passing described in Rules 3 and 4 is always do-able.

\begin{proof}[Proof of Lemma~\ref{lmm:bigk-inv}]
	Both invariants hold on $v$ right after the subtree of $v$ has been reconstructed because at this moment either (i) $|\pilot(v)| = B$, or (ii) $|\pilot(v)| < B$ and meanwhile all proper descendants of $v$ in {\bm$T$} have empty pilot sets.
	
	\vgap
	
	Inductively, assuming that the invariants are valid currently, next we will prove that they remain valid after applying our update algorithms. 
    \begin{itemize} 
        \item Putting a newly inserted point $p$ into $\pilot(v)$ gives $v$ a new insertion token, which accounts for the increment of $|\pilot(v)| - B$. Hence, Invariant 1 still holds. Invariant 2 also holds because $B - |\pilot(v)|$ has decreased. 
        
        \item Physically deleting a point $p \in \pilot(v)$ from {\bm$T$} gives $v$ a new deletion token, which accounts for the increment of $B - |\pilot(v)|$. Hence, Invariant 2 still holds. Invariant 1 also holds because $|\pilot(v)| - B$ has decreased.
        
        \item Consider a push-down at node $v$. After the push-down, $|\pilot(v)| = B$; thus, Invariants 1 and 2 trivially hold on $v$. Let $v'$ be a child of $v$. Invariant 1 still holds on $v'$ because $v'$ gains as many insertion tokens as the increase of $|\pilot(v')|-B$. Invariant 2 also continues to hold on $v'$ because the value of $B - |\pilot(v')|$ has decreased. 
        
        \item Consider a pull-up at node $v$. After the pull-up, $|\pilot(v)| \le B$; hence, Invariant 1 trivially holds on $v$. Invariant 2 also holds on $v$ because $v$ loses as many deletion tokens as the decrease of $B - |\pilot(v)|$. Let $v'$ be a child of $v$. Invariant 1 continues to hold on $v'$ because the value of $|\pilot(v')| - B$ has decreased. Invariant 2 also holds on $v'$ because $v'$ gains as many deletion tokens as the increase of $B - |\pilot(v')|$. 
    \end{itemize}
\end{proof}

Recall that a push-down is necessitated at a node $v$ only if $\pilot(v) > 2B$. Therefore, by Invariant 1, after the operation $|\pilot(v)| - B = \Omega(|\pilot(v)|)$ insertion tokens must have descended to the next level of {\bm$T$}. The operation itself takes $O(|\pilot(v)|/B)$ I/Os; after amortization, each of those insertion tokens bears only $O(1/B)$ I/Os of that cost. 

\vgap

Now consider the moment when a pilot set underflow happens at $v$. By Invariant 2, $v$ must be holding at least $B/2$ deletion tokens at this time. Our algorithm performs one or two pull-ups at $v$ using $O(1)$ I/Os. We account for such cost as follows. If neither of the two pull-ups is a draining one, at least $B/2$ deletion tokens must have descended to the next level; we charge the cost on those tokens, each of which bears $O(1/B)$ I/Os. On the other hand, if a draining pull-up occurred, at least $B/2$ deletion tokens must have disappeared; each of them is asked to bear $O(1/B)$ I/Os. 

\vgap 

In summary, each token before its disappearance is charged $O(\fr{1}{B} \lg n)$ I/Os in total. Since an update creates only one token, the amortized update cost only needs to increase by $O(\fr{1}{B} \lg n)$ to cover the cost of push-downs and pull-ups. 

\extraspacing {\bf Remark.} The above analysis has assumed that the height of $T$ remains $\Theta(\lg n)$. This assumption can be removed by the standard technique of global rebuilding. With this, we have completed the proof of Lemma~\ref{lmm:bigk-main}. 

\section{A Structure for {\large \bm$k = O(\polylg n)$}} \label{sec:smallk}

In this section, we will prove: 

\begin{lemma} \label{lmm:smallk-smallk}
    For top-$k$ range reporting with $k = O(\polylg n)$, there is a structure of $O(n/B)$ space that answers a query in $O(\lg_B n + k/B)$ I/Os, and supports an insertion and a deletion in $O(\lg_B n)$ I/Os amortized.
\end{lemma}

As explained in Section~\ref{sec:intro-ours}, Theorem~\ref{thm:main-final} follows from the combination of Lemmas~\ref{lmm:bigk-main} and \ref{lmm:smallk-smallk}, and a structure of \cite{st12}. To prove Lemma~\ref{lmm:smallk-smallk}, we will first introduce two relevant problems in Sections~\ref{sec:smallk-rank} and \ref{sec:smallk-group}. Our final structure---presented in Section~\ref{sec:smallk-str}---is built upon solutions to those problems.

\subsection{Approximate Union-Rank Selection} \label{sec:smallk-rank} 

Let $L$ be a set of real values. Given a real value $e$, we define its {\em rank} in $L$ as $|\{e' \in L \mid e' \ge e\}|$. Note that the largest element of $L$ has rank 1. 

\vgap

In {\em approximate union-rank selection} (AURS), we are given $m$ disjoint sets $L_1, ..., L_m$ of real values, such that each $L_i$ ($1 \le i \le m$) can be accessed only by the following operators:
\begin{itemize}
    \item \textsc{Max}: Returns the largest element of $L_i$ in $cost_{max}$ I/Os.

	\item \textsc{Rank}: Given a real-valued parameter $\rho \in [1, \fr{1}{c_1}|L_i|]$ where $c_1 \ge 2$ is a constant, this operator returns in $cost_{rank}$ I/Os an element $e \in L_i$ whose rank in $L_i$ falls in $[\rho, c_1 \rho)$.
\end{itemize}
Given an integer $k$ satisfying 
\begin{eqnarray}
    1 \le k \le (1/c_1) \cdot \min\{|L_1|, ..., |L_m|\},
    \label{eqn:smallk-rankcond}
\end{eqnarray}
a query returns an element $e \in \bigcup_{i=1}^m L_i$ whose rank in $\bigcup_{i=1}^m L_i$ falls in $[k, c'k]$, where $c' > 1$ is a constant dependent only on $c_1$.

\vgap 

AURS is reminiscent of a rank selection problem defined by Frederickson and Johnson \cite{fj82}. However, their algorithm assumes a more powerful \textsc{Rank} operator that returns an element in $L_i$ with a {\em precise} rank. In the appendix, we show how to adapt their algorithm to obtain the result below: 

\begin{lemma} \label{lmm:smallk-aurs}
    Each query in the AURS problem can be answered in $O(m (cost_{max} + cost_{rank}))$ I/Os.
\end{lemma}

\subsection{Approximate {\large \boldmath$(f, l)$}-Group {\large \boldmath$k$}-Selection} \label{sec:smallk-group}

Given integers $f$ and $l$, we define an {\em $(f, l)$-group} {\bm$G$} as a sequence of $f$ disjoint sets $G_1, ..., G_f$, where each $G_i$ ($1 \le i \le f$) is a set of at most $l$ real values. Let $N \ge fl$ be  an integer such that a word has $\Omega(\lg N)$ bits. 

\vgap

In the {\em approximate $(f, l)$-group $k$-selection problem}---henceforth, the {\em $(f, l)$-problem} for short---the input is an $(f, l)$-group {\bm$G$}, where the values of $f, l$, $N$, and $B$ (block size) satisfy all of the following:
\begin{itemize} 
	\item $l = O(\polylg N)$ 
	\item $f \le \sqrt{B} \lg^\eps N$ where $\eps$ is a constant satisfying $0 < \eps < 1$. 
\end{itemize}
A query is given:
\begin{itemize} 
    \item an interval $q = [\alpha_1, \alpha_2]$ with $1 \le \alpha_1 \le \alpha_2 \le f$,
    \item and a real value $k \in [1, |\bigcup_{i\in q} G_i|]$;
\end{itemize}
it returns a real value $x$ whose rank in $\bigcup_{i\in q} G_i$ falls in $[k, c_2 k)$, where $c_2 \ge 2$ is a constant. It is required that $x$ should be either $-\infty$ or an element in $\bigcup_{i\in q} G_i$. 

\vgap 

The following lemma is a crucial result that stands at the core of our final structure. Its proof is non-trivial and delegated to Section~\ref{sec:flprob}. 

\begin{lemma} \label{lmm:smallk-fl-prob}
	For the $(f, l)$-problem, we can store {\bm$G$} in a structure of $O(fl/B)$ space that answers a query in $O(\lg_B (fl))$ I/Os, and supports an insertion and deletion in $O(\lg_B (fl))$ I/Os amortized. 
\end{lemma}

\subsection{Proof of Lemma~\ref{lmm:smallk-smallk}} \label{sec:smallk-str}

We are now ready to elaborate on the structure claimed in Lemma~\ref{lmm:smallk-smallk}. It suffices to focus on the {\em approximate range $k$-selection problem}: 

\begin{center}
\begin{minipage}{0.9\linewidth}
    The input is the same set $S$ of points as in top-$k$ range reporting. Given an interval $q = [x_1, x_2]$ and an integer $k$ satisfying $1 \le k \le |S \intr q|$, a query returns a point $e \in S \intr q$ such that between $k$ and $O(k)$ points in $S \intr q$ have scores at least $\score(e)$. 
\end{minipage}
\end{center}
Suppose that there is a structure solving the above problem with query time $t_q$ and amortized update time $t_u$. Then, we immediately obtain a structure of asymptotically the same space for top-$k$ range reporting with query time $O(t_q + \lg_B n + k/B)$ and amortized update time $O(t_u + \lg_B n)$ (see \cite{st12}). A structure with $t_q = \lg_B n$ and $t_u = \lg^2_B n$ was given in \cite{st12}. 

\vgap 

Fix an integer $l = O(\polylg n)$. Next, assuming $k \le l$, we describe a linear-size structure with $t_q = t_u = O(\lg_B n)$, which therefore yields a structure of Lemma~\ref{lmm:smallk-smallk}. 

\extraspacing {\bf Structure.} We build a WBB-tree $T$ on $S$ with branching parameter $f = \sqrt{B\lg n}$, and leaf capacity $b = flB$. Each node $u$ naturally corresponds to an x-range in $\real$. If $u$ is an internal node with child nodes $u_1, ..., u_f$, define a {\em multi-slab} to be the union of the x-ranges of $u_i, u_{i+1}, ..., u_j$ for some meaningful $i, j$. 

\vgap 

Given an (internal/leaf) node $u$, let $S_u$ be the set of elements stored in the subtree of $u$. Define $G_u$ as the set of $c_2 l$ highest scores of the elements in $S_u$, where $c_2$ is the constant mentioned in the definition of the $(f, l)$-problem in Section~\ref{sec:smallk-group}. 

\vgap

For each leaf node $u$, maintain a structure of \cite{st12} to support approximate range $k$-selection on $S_u$. Consider now $u$ as an internal node with child nodes $u_1, ..., u_f$. We
\begin{itemize} 
    \item maintain an {\em $(f, c_2 l)$-structure} of Lemma~\ref{lmm:smallk-fl-prob} on the $(f, c_2 l)$-group {\boldmath$G$}$_u = (G_{u_1}, ..., G_{u_f})$, with $N$ fixed to some integer in $[n, 4n]$ (this will be guaranteed by our update algorithms). 
    
    \item store $G_{u_1} \cup ... \cup G_{u_f}$ in a (slightly augmented) B-tree so that, for any $1 \le \alpha_1 \le \alpha_2 \le f$, the maximum score in $\bigcup_{i \in [\alpha_1, \alpha_2]} G_{u_i}$ can be found in $O(\lg_B (fl))$ I/Os.
\end{itemize}

There are $O(n/(fb))$ internal nodes, each of which occupies $O(fl/B)$ blocks. Hence, all the internal nodes use altogether $O(\fr{n}{fb} \cdot \fr{fl}{B}) = o(n/B)$ space. The overall space cost is therefore $O(n/B)$. 

\extraspacing {\bf Query.} Given a query with parameters $q = [x_1, x_2]$ and $k$, search $T$ in a standard way to identify a minimum set $C$ of $O(\lg_f n)$ disjoint {\em canonical ranges} whose union covers $q$, such that each canonical range is either the x-range of a leaf node or a multi-slab. 

\vgap 

Define $S_m = m \intr S$ for each multi-slab $m \in C$. Perform AURS with parameter $k$ on $\{S_m \mid m \in C\}$. At each internal node $u$ on which a multi-slab $m \in C$ is defined, the $(f, c_2 l)$-structure of $u$ and the B-tree on {\boldmath$G$}$_u$ allow us to implement the \textsc{Rank} and \textsc{Max} operators on $S_m$ in $O(\lg_B (fl))$ I/Os, respectively. Therefore, by Lemma~\ref{lmm:smallk-aurs}\footnote{The constant $c_1$ in the \textsc{Rank} operator's definition (see Section~\ref{sec:smallk-rank}) equals $c_2$ here, as is guaranteed by the $(f,c_2 l)$-structures. Given that we focus on $k \le l$ while each $G_u$ has size $c_2 l$, we know that the condition stated in \eqref{eqn:smallk-rankcond} always holds.}, the AURS finishes in $O(\lg_f n \cdot \lg_B (fl)) = O(\lg_B n)$ I/Os. Denote by $e$ the element returned\footnote{The AURS returns only the score of $e$, but it is easy to fetch $e$ by the score in $O(\lg_B n)$ I/Os.}.

\vgap 

For each leaf node $z$ whose x-range is in $C$, perform approximate range $k$-selection on $S_z$ using $q$ in $O(\lg_B b) = O(\lg_B n)$ I/Os. There are at most two such leaf nodes; let $e_1, e_2$ be the 
results of approximate range $k$-selection on them, respectively. We return $\max\{e, e_1, e_2\}$ as the final answer.

\extraspacing {\bf Update.} The update algorithm (which is relatively standard) can be found in the appendix.

\section{Solving the \bm{\large $(f,l)$}-PROBLEM} \label{sec:flprob}

We devote this section to proving Lemma~\ref{lmm:smallk-fl-prob}. Henceforth, by ``query'', we refer to a query in the $(f, l)$-problem. When no ambiguity can arise, we use {\bm$G$} to denote also the union of $G_1$, ... $G_f$. 

\subsection{A Static Structure} \label{sec:flprob-sta} 

We will need a tool called the {\em logarithmic sketch}---henceforth, {\em sketch}---developed in \cite{st12}. Let $L$ be a set of $l$ real values. Its sketch $\Sigma$ is an array of size $\lf \lg l \rf + 1$, where the $j$-th ($1 \le j \le \lf \lg l \rf + 1$) entry $\Sigma[j]$---called a {\em pivot}---is an element in $L$ whose rank in $L$ falls in $[2^{j-1}, 2^j)$; any such element can be used as $\Sigma[j]$. 
 
\begin{lemma} [\cite{st12}] \label{lmm:smallk-st12}
     Let $L_1, ..., L_m$ be $m$ disjoint sets of real values. Given their sketches and a real value $k$ satisfying $1 \le k \le |\bigcup_{i=1}^m L_i|$, we can find in $O(m)$ I/Os a real value $x$ whose rank in $\bigcup_{i=1}^m L_i$ is between $k$ and $c_3 k$ (where $c_3 \ge 2$ is a constant). Furthermore, $x$ is either $-\infty$ or an element in $\bigcup_{i=1}^m L_i$.
\end{lemma}

Create a sketch $\Sigma_i$ for each $G_i$ ($1 \le i \le f$). Call the set $\{\Sigma_1, ..., \Sigma_f\}$ a {\em sketch set}. We store a compressed form of the sketch set as follows. Describe each pivot $e \in \Sigma_i$ by its {\em global rank} in {\bm$G$} using $\lg(fl)$ bits, and by its {\em local rank} in $G_i$ using $\lg l$ bits. Hence, each $\Sigma_i$ requires $\lg l \cdot 2\lg(fl)$ bits. A compressed sketch set occupies $f \lg l \cdot 2\lg(fl) = \sqrt{B} \cdot \lg^{\eps} N \cdot O((\lg\lg N)^2)$ bits, and thus fits in a block (which has $B \cdot \Omega(\lg N)$ bits).  

\vgap

Given a query, we first spend an I/O reading the compressed sketched set, and then run the algorithm of Lemma~\ref{lmm:smallk-st12} on it in memory. Suppose that this algorithm outputs $x$. If $x = -\infty$, we simply return $-\infty$ as our final answer. Otherwise, $x$ is equal to the global rank of an element in {\bm$G$}. To convert the global rank to an actual element, we index all the elements of {\bm$G$} with a B-tree, which supports such a conversion in $O(\lg_B (fl))$ I/Os. The overall space is $O(fl/B)$ (due to the B-tree); and the query cost is $O(\lg_B (fl))$. Notice that the constant $c_2$ in Section~\ref{sec:smallk-group} equals the constant $c_3$ stated in Lemma~\ref{lmm:smallk-st12}. 

\subsection{Supporting Insertions} \label{sec:flprob-ins} 

To facilitate updates, we store the elements of each $G_i$ ($1 \le i \le f$) in a B-tree that allows us to obtain the element of any specific local rank in $O(\lg_B l)$ I/Os. In addition, we also maintain a structure of the following lemma, whose proof is deferred to Section~\ref{sec:flprob-rank}: 

\begin{lemma} \label{lmm:flprob-rank}
	We can store an $(f, l)$-group $\G = (G_1, ..., G_f)$ in a structure of $O(fl/B)$ space such that, in one I/O, we can read into memory a single block, from which we can obtain for free the global rank of the element with local rank $r$ in $G_i$, for every $r \in [1, \sqrt{B} \lg_B (fl)]$ and every $i \in [1, f]$. The structure supports an insertion and a deletion in $O(\lg_B (fl))$ I/Os. 
\end{lemma}

Suppose that an element $e_{new}$ is to be inserted in $G_i$ for some $i \in [1, f]$. Let $r_{new}$ be the rank of $e_{new}$ in {\bm$G$}. We observe that, except perhaps a single pivot, the new compressed sketch set (after the update) can be deduced from: the current compressed sketch set, $r_{new}$ and $i$. To understand this, consider first a compressed sketch $\Sigma_{i'}$ where $i' \neq i$. Each pivot whose global rank is at least $r_{new}$ now has its global rank increased by 1 (its local rank is unaffected). Regarding the compressed $\Sigma_i$, the same is true, but additionally every such pivot should also have its local rank increased by 1. Furthermore, a new pivot is needed in $\Sigma_i$ if $|G_i|$ reaches a power of 2 after the insertion---in such a case we say that $\Sigma_i$ {\em expands}; the new pivot is the only one in the compressed sketch set that cannot be deduced (because its global rank is unknown).

\vgap

Motivated by this observation, to insert $e_{new}$ in $G_i$, we first obtain $r_{new}$ from the B-tree of {\bm$G$} in $O(\lg_B (fl))$ I/Os, and then update the new compressed sketch set as described earlier in 1 I/O. Next, $e_{new}$ is inserted in the B-trees of {\bm$G$} and $G_i$ using $O(\lg_B(fl))$ I/Os. If now $|G_i|$ is a power of 2, we retrieve the global rank of the smallest element in $G_i$ in $O(\lg_B (fl))$ I/Os, and add the element to $\Sigma_i$ in memory.

\vgap

Recall that, the $j$-th ($1 \le j \le \lf \lg l \rf + 1$) pivot of $\Sigma_i$ should have its local rank confined to $[2^{j-1}, 2^j)$. If this is not true, we say that it is {\em invalidated}. The insertion may have invalidated one or more pivots, (all of which can be found with no I/O because $\Sigma_i$ in memory). Upon the invalidation of $\Sigma_i[j]$, we replace it as the element $e \in G_i$ with local rank $\lf \fr{3}{2} \cdot 2^{j-1} \rf$ so that $\Omega(2^j)$ updates in $G_i$ are needed to invalidate $\Sigma_i[j]$ again. For the replacement to proceed, it remains to obtain the global rank of $e$. We do so by distinguishing two cases: 

\begin{itemize} 
    \item {\em Case $2^j \ge \sqrt{B} \lg_B (fl)$.} We simply fetch $e$ from the B-tree on $G_i$, and obtain its global rank from the B-tree on {\bm$G$}. We can now update $\Sigma_i[j]$ in memory.  

	\vgap

	In total, the invalidated pivot is fixed with $O(\lg_B (fl)) = O(2^j/\sqrt{B})$ I/Os. Since $\Omega(2^j)$ updates must have occurred in $G_i$ to trigger the invalidity of $\Sigma_i[j]$, each of those updates accounts for $O(1/\sqrt{B})$ I/Os of the pivot recomputation. As an update can be charged at most $O(\lg l)$ times this way (i.e., once for every $j \ge \sqrt{B} \lg_B (fl)$), its amortized cost is increased by only $O(\fr{1}{\sqrt{B}} \lg l) = O(\lg_B l)$.
	
	\item {\em Case $2^j < \sqrt{B} \lg_B (fl)$.} There are $O(\lg (\sqrt{B}\lg_B (fl))$ such invalidated pivots in $\Sigma_i$. We can recompute {\em all} of them together in $O(1)$ I/Os using Lemma~\ref{lmm:flprob-rank}. 
\end{itemize}

Overall, an insertion requires $O(\lg_B (fl))$ I/Os amortized. 

\subsection{Supporting Deletions} \label{sec:flprob-del} 

Suppose that an element $e_{old}$ is to be deleted from $G_i$ for some $i \in [1, f]$. Let $r_{old}$ be the rank of $e_{old}$ in {\bm$G$}. Except possibly for only one pivot, the new compressed sketch set can be deduced based only on the current compressed sketch set, $r$, and $i$. To see this, consider first $\Sigma_{i'}$ where $i' \neq i$. Each pivot whose global rank is larger than $r_{old}$ now needs to have its global rank decreased by 1. Regarding $\Sigma_i$, the same is true, and every such pivot should also have its local rank decreased by 1. Furthermore, the last pivot of $\Sigma_i$ should be discarded if $|G_i|$ was a power of 2 before the deletion: in such a case, we say that $\Sigma_i$ {\em shrinks}. Finally, if $e_{old}$ happens to be a pivot of $\Sigma_i$, a new pivot needs to be computed to replace it---this is the only pivot that cannot be deduced; we call it a {\em dangling} pivot.

\vgap

The concrete steps of deleting $e_{old}$ are as follows. After fetching its global rank $r_{old}$ in $O(\lg_B (fl))$ I/Os, we update the compressed sketch set in memory according to the above discussion. If $\Sigma_i$ shrinks, we delete the last pivot $\Sigma_i$ in memory. If $e_{old}$ was a pivot (say, the $j$-th one for some $j$), we retrieve the element $e$ with local rank $\lf \fr{3}{2} \cdot 2^{j-1} \rf$ in $G_i$, and obtain its global rank using $O(\lg_B (fl))$ I/Os. We then replace the dangling pivot $\Sigma_i[j]$ with $e$ in memory. 

\vgap 

Finally, recompute the invalidated pivots (if any) in the same way as in an insertion. As analyzed in Section~\ref{sec:flprob-ins}, such recomputation increases the amortized update cost by only $O(\lg_B l)$.

\subsection{Proof of Lemma~\ref{lmm:flprob-rank}} \label{sec:flprob-rank}

Let us define the list of the $\sqrt{B} \lg_B (fl)$ largest elements of $G_i$ ($1 \le i \le f$) as the {\em prefix} of $G_i$, and denote it as $P_i$. Let {\boldmath$P$} be the union of $P_1, .., P_f$; we refer to {\boldmath$P$} as a {\em prefix set}. {\boldmath$P$} contains at most $f \sqrt{B} \lg_B (fl)$ points. 

\vgap 

We compress {\boldmath$P$} by describing each element $e$ (say, $e \in P_i$ for some $i$) in {\boldmath$P$} using its global rank in {\bm$G$} and its local rank in $G_i$, for which purpose $O(\lg (fl))$ bits suffice. Hence, {\boldmath$P$} can be described by $f \cdot \sqrt{B} \lg_B (fl) \cdot O(\lg(fl)) = \sqrt{B} \lg^{\eps} N \cdot \sqrt{B} \cdot O((\lg\lg N)^2) = B \cdot \lg^{\eps} N \cdot O((\lg\lg N)^2)$ bits, which fit in a block. After loading this block into memory, we can obtain the global rank of the $r$-th largest element of $G_i$ for free, regardless of $r \in [1, \sqrt{B}\lg_B (fl)]$ and $i$. 

\vgap 

Besides the aforementioned block, we also maintain a B-tree on each $G_i$ ($1 \le i \le f$) and a B-tree on {\bm$G$}. The space consumed is $O(fl/B)$.

\extraspacing {\bf Insertion.} Suppose that we need to insert an element $e_{new}$ into $G_i$. First, we update the B-trees of $G_i$ and {\bm$G$} in $O(\lg_B (fl))$ I/Os. With the same cost, we can also decide whether $e_{new}$ should enter $P_i$. If not, the insertion is complete.  

\vgap 

Otherwise, we find the global rank $r_{new}$ of $e_{new}$ and its local rank $r_{new}'$ in $G_i$ with $O(\lg_B (fl))$ I/Os. Load the compressed prefix set {\bm$P$} into memory with 1 I/O. Then, the new compressed prefix set can be determined for free based on {\boldmath$P$}, $i$, $r_{new}$, and $r_{new}'$. To see this, first consider a compressed prefix $P_{i'}$ with $i' \neq i$: if an element has global rank at least $r_{new}$, it should have its global rank increased by 1. Regarding the compressed prefix $P_i$, the same is true; furthermore, all such elements in $P_i$ should also have their local ranks increased by 1. Finally, we add $e_{new}$ into $P_i$; if $P_i$ has a size over $\sqrt{B}\lg_B (fl)$, we discard its smallest element.   

\extraspacing {\bf Deletion.} Suppose that we need to delete an element $e_{old}$ from $G_i$. Using the B-tree on {\bm$G$}, we find its global rank $r_{old}$ in $O(\lg_B (fl))$ I/Os. Then, $e_{old}$ is removed from the B-trees of $G_i$ and {\bm$G$} in $O(\lg_B (fl))$ I/Os. 

\vgap 

If $e_{old} \notin P_i$, the deletion is done. Otherwise, we load the compressed prefix set in 1 I/O, and then update it, except for a single element, in memory. Specifically, in a compressed prefix $P_{i'}$ with $i' \neq i$, if an element has global rank at least $r_{old}$, it should have its global rank decreased by 1. Regarding the compressed prefix $P_i$, the same is true; furthermore, all such elements in $P_i$ should also have their local ranks decreased by 1. 

\vgap

The last element of $P_i$ is the only one that cannot be inferred directly at this point. But it can be filled in simply by retrieving the element with local rank $\sqrt{B} \lg_B (fl)$ in $G_i$, and then its global rank in {\bm$G$}, all in $O(\lg_B (fl))$ I/Os.

%

\bibliographystyle{abbrv}
\bibliography{ref}


\section*{Appendix}
\section*{Proof of Lemma~\ref{lmm:smallk-aurs}}

In this proof, set $c = c_1$ and $L = \cup_{i=1}^m L_i$. Given an element $e \in L_i$ ($1 \le i \le m$), we refer to its rank in $L_i$ as its {\em local rank}, and its rank in $L$ as its {\em global rank}. 

\extraspacing {\bf Case \bm$k \ge m$.} Our algorithm executes in $\lc \lg_c m \rc$ rounds. In the $j$-th round ($1 \le j \le \lc \lg_c m \rc$), $\lc m/c^{j-1} \rc$ sets among $L_1, ..., L_m$ are {\em active}, while the others are {\em inactive}. At the beginning, $L_1, ..., L_m$ are all active. 

\vgap 

In round $j \in [1, \lc \lg_c m \rc]$, we execute \textsc{Rank} on each active set $L_i$ with parameter $\rho = c^j k/m$. Remember that the operator can return any element whose local rank falls in $[c^j k/m, c^{j+1} k/m)$.\footnote{Such an element definitely exists because $c^j k/m \le ck \le |L_i|$.} Let $P'$ be the set of elements fetched. We call each element in $P'$ a {\em marker}, and assign it a {\em weight} equal to 
\begin{itemize} 
    \item $\lc c k/m \rc$ if $j = 1$;
    \item $\lc c^j k/m \rc - \lc c^{j-1} k/m \rc$ if $j > 1$. 
\end{itemize}
The $\lc m/c^j \rc$ largest markers in $P'$ are taken as {\em pivots}, among which the smallest is the {\em cutoff pivot} of this round. An active set remains active in the next round if its marker 
is a pivot, whereas the other active sets become inactive.

\vgap

Denote by $P_j$ the set of pivots taken in the $j$-th round ($1 \le j \le \lc \lg_c m \rc]$), and by $P$ the union of $P_1, P_2, ..., P_{\lc \lg_c m \rc}$. It is clear that $|P| = \sum_{j=1}^{\lc \lg_c m \rc} \lc m/c^j \rc = O(m)$. 

\vgap

Consider a pivot $p \in P_j$, and suppose that it comes from $L_i$ for some $i \in [1, m]$. Define the {\em local prefix weight} of $p$ as the sum of the weights of the first $j$ pivots fetched from $L_i$. By how we define weights, it is easy to verify that the local prefix weight of $p$ is $\lc c^jk/m \rc$. For each pivot $p \in P$, define its {\em prefix weight} as the total weight of all the pivots that are larger than or equal to $p$.

\begin{observation} \label{obs:app-aurs-cutoff}
	Every cutoff pivot has a prefix weight at least $k$.
\end{observation}

\begin{proof}
	Consider the cutoff pivot $p^\star$ of round $j \in [1, \lc \lg_c m \rc]$. Let $P_j'$ be the set of $\lc m/c^j \rc$ pivots of $P_j$ greater than or equal to $p^\star$. The prefix weight of $p^\star$ is at least the sum of the local prefix weights of all the pivots in $P_j'$, which is at least $\lc m/c^j \rc \lc c^j k / m \rc \ge k$.
\end{proof}

We perform a {\em weighted selection} to find the largest pivot $v \in P$ whose prefix weight is at least $k$ ($v$ definitely exists by the previous observation). The algorithm terminates by returning $v$.

\vgap 

The algorithm performs in $O(m\cdot cost_{rank})$ I/Os because the $j$-th round takes $O((m/c^{j-1}) \cdot cost_{rank})$ I/Os (i.e., geometrically decreasing with $j$), while the weighted selection needs only $O(m/B)$ I/Os. Next, we prove that the algorithm is correct, namely, the global rank of $v$ is in $[k, c'k]$ for some constant $c'$ dependent only on $c$.

\begin{observation} \label{obs:app-aurs-prefixweight-of-v}
    The prefix weight of $v$ is at most $(1+2c) k$. 
\end{observation}

\begin{proof}
    We define $v'$ as the smallest pivot in $P$ that is larger than $v$. By definition, we know that the prefix weight of $v'$ is smaller than $k$. Define $p_i'$ as the smallest pivot in $L_i$ that is larger than or equal to $v'$; $p_i' = \nil$ if no such pivot exists. Defining the local prefix weight of a $\nil$ point to be 0, we have: 
    \begin{eqnarray} 
        \textrm{prefix weight of $v'$} &=& \sum_{i=1}^m \textrm{local prefix weight of $p_i'$} \nn \\
        &<& k. \label{eqn:app-aurs-prefixweight-of-v-eq1}
    \end{eqnarray}
    Clearly, it also holds that
    \begin{eqnarray} 
        \textrm{prefix weight of $v$} &=& \textrm{prefix weight of $v'$} + \textrm{weight of $v$} \nn \\ 
        &<& k + \textrm{weight of $v$} \label{eqn:app-aurs-prefixweight-of-v-eq2}
    \end{eqnarray}
    Let $i^*$ be such that $v \in L_{i^*}$. We distinguish two cases: 
    
    \vgap 
    
    {\em Case 1: $p'_{i^*} = \nil$.} This means that $v$ was taken in the first round of our algorithm; hence, its weight is $\lc ck/m \rc < 2ck$ (recall that $c \ge 2$). Therefore, by \eqref{eqn:app-aurs-prefixweight-of-v-eq2} the prefix weight of $v$ is less than $k + 2ck$. 
    
    \vgap 
    
    {\em Case 2: $p'_{i^*} \neq \nil$.} Then, the weight of $v$ is less than $2c$ times the local prefix weight of $p'_{i^*}$. Together with \eqref{eqn:app-aurs-prefixweight-of-v-eq1}, this implies that the weight of $v$ is less than $2c k$. Therefore, once again, \eqref{eqn:app-aurs-prefixweight-of-v-eq2} tells us that the prefix weight of $v$ is less than $k + 2ck$. 
\end{proof}

By Observation~\ref{obs:app-aurs-cutoff} and the definition of $v$, all cutoff pivots are smaller than or equal to $v$. Hence, every $L_i$ has at least one marker smaller than or equal to $v$. We will refer to the largest such marker as the {\em succeeding marker} of $L_i$, and denote it as $e_i$. Note that $e_i$ is not necessarily a pivot.

\vgap

Let $p_i$ be the smallest pivot of $L_i$ that is larger than or equal to $v$. If $p_i$ exists, we say that $L_i$ is {\em pivotal}; otherwise, $L_i$ is {\em non-pivotal}. For a pivotal $L_i$, define:
\begin{itemize}
	\item $r_i =$ the local prefix weight of $p_i$. 

	\item $S_i =$ the set of elements in $L_i$ that are larger than or equal to $v$.
	
	\item $S_i' =$ the set of elements in $L_i$ that are larger than or equal to $p_i$.
\end{itemize}

\begin{observation} \label{obs:app-size}
	$r_i \le |S_i'| \le |S_i| < c^2 \cdot r_i$.
\end{observation}
\begin{proof}
	Let $j$ be such that $p_i$ was taken in the $j$-th round of our algorithm. $|S_i'|$ is exactly the local rank of $p_i$, which must fall in $[c^j k / m, c^{j+1}k/m)$. Hence, it follows that $|S_i'| \ge \lc c^jk/m\rc = r_i$. 

	\vgap 
	
	As $e_i$ was taken in the $(j+1)$-st round, its local rank is less than $c^{j+2}k/m$. Since the local rank $e_i$ is an upper bound of $|S_i|$, it follows that $|S_i| < c^{j+2}k/m \le c^2 r_i$. 
\end{proof}

In the pivotal sets, the total number of elements larger than or equal to $v$ equals:
\begin{eqnarray}
	\sum_{\textrm{$i$ s.t.\ $L_i$ is pivotal}} |S_i|
	&\le&
	c^2 \sum_{\textrm{$i$ s.t.\ $L_i$ is pivotal}} r_i \nn \\
	&=&
	c^2 \cdot (\textrm{prefix weight of $v$}) \nn \\
	\textrm{(by Observation~\ref{obs:app-aurs-prefixweight-of-v})} &<&
	c^2(1+2c) k. \nn
\end{eqnarray}
Each non-pivotal set $L_i$ has less than $c^2 k/m$ elements larger than or equal to $v$. It thus follows that the global rank of $v$ is less than $(c^2 k / m) m + c^2(1+2c) k = c^2 (2+2c) k$.

\vgap 

On the lower side, the global rank of $v$ is at least 
\begin{eqnarray} 
    \sum_{\textrm{$i$ s.t.\ $L_i$ is pivotal}} |S_i|
    &\ge& 
    \sum_{\textrm{$i$ s.t.\ $L_i$ is pivotal}} r_i \nn \\
    &=&
	\textrm{prefix weight of $v$} \nn \\
	&\ge& k. \nn 
\end{eqnarray}

\extraspacing {\bf Case \bm$k < m$.} From each $L_i$, we use \textsc{Max} to request the largest element in $L_i$. Let $P'$ be the set of elements fetched (i.e., one from each $L_i$). Obtain the $k$-th largest element $v'$ in $P'$. Make set $L_i$ inactive if its largest element is smaller than $v'$; otherwise, $L_i$ is active. Run the above algorithm on the $k$ active sets. Suppose that the algorithm outputs $v$. We then return $\max\{v, v'\}$ as the final answer. It is easy to prove that the algorithm is correct, and runs in $O(m (cost_{rank} + cost_{max}))$ I/Os.


\section*{The Update Algorithm in Section~\ref{sec:smallk-str}}

To support updates, for each internal node $u$, build a B-tree on the scores in {\bm$G$}$_u$. For each leaf node $z$, build a B-tree on the scores of the elements in $S_z$. Refer to these B-trees as {\em score B-trees}. Denote by $parent(u)$ the parent of $u$.

\vgap

To insert a point $e$ in $S$, first descend a root-to-leaf path $\pi$ to the leaf node $z$ whose x-range covers $e$. At $z$, update all its secondary structures in $O(\lg^2_B b) = O((\lg_B \lg n)^2) = O(\lg_B n)$ amortized I/Os. Next, we fix the secondary structures of the nodes along $\pi$ in a bottom up manner. If $\score(e)$ enters $G_z$, at $parent(z)$, delete the lowest score in $G_z$, and then insert $\score(e)$ in $G_z$. The secondary structures of $parent(z)$ are then updated accordingly. In general, after updating an internal node $u$, we check using the score B-tree of $u$ whether $\score(e)$ should enter $G_u$. If so, at $parent(u)$, delete the lowest score in $G_u$, insert $\score(e)$ in $G_u$, and update the secondary structures of $parent(u)$. By Lemma~\ref{lmm:smallk-fl-prob}, we spend $O(\lg_B (fl))$ amortized I/Os at each node, and hence, $O(\lg_B n)$ amortized I/Os in total along the whole $\pi$.

\vgap

We now explain how to handle node splits. Suppose that a leaf node $z$ splits into $z_1, z_2$. First, build the secondary structures of $z_1$ and $z_2$ in $O(b \lg^2_B b)$ I/Os. At $v = parent(z)$, destroy $G_z$, and include $G_{z_1}$ and $G_{z_2}$ into {\bm$G$}$_v$. Rebuild all the secondary structures at $v$ in $O(fl \cdot \lg_B (fl)) = O(b \lg_B b)$ I/Os (Lemma~\ref{lmm:smallk-fl-prob}). This cost can be amortized over the $\Omega(b)$ updates that must have taken place in $z$, such that each update is charged only $O(\lg^2_B b) = O(\lg_B n)$ I/Os.

\vgap

A split at an internal level can be handled in a similar way. Suppose that an internal node $u$ splits into $u_1, u_2$. Divide {\bm$G$}$_{u}$ into {\bm$G$}$_{u_1}$ and {\bm$G$}$_{u_2}$ in $O(fl/B)$ I/Os, and then rebuild the secondary structures of $u_1, u_2$ in $O(fl \cdot \lg_B (fl))$ I/Os. After discarding $G_u$ but including $G_{u_1}, G_{u_2}$, we rebuild the secondary structures of $parent(u)$ in $O(fl \cdot \lg_B (fl))$ I/Os. On the other hand, $\Omega(fl)$ updates must have taken place in the subtree of $u$ (recall that the base tree is a WBB-tree). Hence, each of those updates bears $O(\lg_B (fl))$ I/Os for the split cost. As an update bears such cost for at most one node per level, the amortized update cost increases by only $O(\lg_B n)$.

\vgap

An analogous algorithm can be used to handle a deletion in $O(\lg_B n)$ amortized I/Os. After $n$ has doubled or halved, we destroy the entire structure, reset $N$ to $2n$, and rebuild everything in $O(n \lg_B n)$ I/Os. The amortized update cost is therefore $O(\lg_B n)$. 

\end{sloppy}
\end{document}

%% file: yf-def.tex
\def\extraspacing{\vspace{2mm} \noindent}
\def\figcapup{\vspace{-1mm}}
\def\figcapdown{\vspace{-0mm}}

\def\vgap{\vspace{0mm}}

\newtheorem{theorem}{Theorem}
\newtheorem{lemma}{Lemma}

\def\bm{\boldmath}

\def\eps{\epsilon}
\def\fr{\frac}
\def\-{\mbox{-}}

\def\real{\mathbb{R}}

\def\lc{\lceil}
\def\lf{\lfloor}
\def\rc{\rceil}
\def\rf{\rfloor}

\def\nn{\nonumber}

\def\*{\star}

\DeclareMathOperator*{\polylg}{polylg}

\DeclareMathOperator*{\intr}{\cap}

